\documentclass[letterpaper,11pt]{article}
\usepackage{amsfonts}
\usepackage[pdftex]{graphicx}
\usepackage{amsmath,amssymb,amsthm}
\usepackage{natbib}
\usepackage{hyperref}
\usepackage{bm}
\usepackage{fancyhdr}
\usepackage{sectsty}
\usepackage[notcite,notref,final]{showkeys}
\usepackage{paralist} 
\usepackage{tikz}
\usepackage{stmaryrd}
\usepackage{setspace}
\usepackage{todonotes}
\usepackage{multirow}
\usepackage{booktabs}
\usepackage{hyperref}
\usepackage{comment}

\hypersetup{colorlinks=true, linkcolor=blue, citecolor=blue}

\def\argmax{\mathop{\rm arg\,max}}
\def\argmin{\mathop{\rm arg\,min}}

\newcommand{\SY}{\Sigma_Y} %Proposal: {\mathcal{B}_Y}
\newcommand{\I}{\mathbb{I}}

\newcommand{\ex}{X} %Proposal: {{\boldsymbol{x}}}
\newcommand{\ey}{Y} %Proposal: {{\boldsymbol{y}}}
\newcommand{\eu}{U} %Proposal: {{\boldsymbol{u}}}
\newcommand{\fq}[1]{\mathfrak{q}_{#1}}

\newcommand{\cP}{\mathcal{P}}
\newcommand{\cW}{\mathcal{W}}
\newcommand{\cX}{\mathcal{X}}  
\newcommand{\cY}{\mathcal{Y}}  
\newcommand{\cU}{\mathcal{U}}  
\newcommand{\cS}{\mathcal{Z}}  
\newcommand{\cM}{\mathcal{M}}

\newcommand{\cC}{\mathcal{C}}

\newcommand{\cL}{\mathcal L}

\newcommand{\f}{F_\theta} 
\newcommand{\G}{G} %Proposal: {\mathsf{m}}  
\newcommand{\qs}[1]{q_{#1}} %Proposal: {\mathsf{q}}  
\newcommand{\contf}{\nu_\theta} %Containment functional %Proposal:{\mathbb{C}_\theta}
\numberwithin{equation}{section}
\theoremstyle{plain}
\newtheorem{theorem}{Theorem} \newtheorem{proposition}{Proposition}[section]  \newtheorem{corollary}{Corollary} \newtheorem{assumption}{Assumption} \newtheorem{remark}{Remark}[section]  
\newtheorem{definition}{Definition}[section]
\theoremstyle{definition}
\newtheorem{example}{Example}

\newcommand{\citeposs}[1]{\citeauthor{#1}'s \citeyearpar{#1}} 

\begin{document}
% \begin{frontmatter}

\title{Universal Inference for Incomplete Discrete Choice Models\thanks{This paper builds on earlier work by the authors titled ``Robust Likelihood-ratio Tests in Incomplete Economic Models'', which evolved into two papers including this paper. We thank Jean-Jacques Forneron for his comments on the earlier work and for pointing us to work on universal inference.

We are grateful to \'{A}ureo de Paula, Hidehiko Ichimura, Toru Kitagawa, Seojeong Lee, Simon Lee, Kirill Ponomarev, Azeem Shaikh, and Alex Torgovitsky for their helpful comments. We thank, for their comments, the seminar, lecture, and conference participants at the following places: Chicago, CyberAgent (AI lab), KER International Conference, Kyoto Summer Workshop in Applied Economics, and EEA-ESEM. We thank Undral Byambadalai, Yan Liu, Junwen Lu, Jiahui Guo and Jie Huang for their excellent research assistance.
Financial support from NSF grants SES-1824344 and SES-2018498 is gratefully acknowledged.   Part of the research was conducted at Microsoft Research (New England) and the University of Tokyo while Kaido was visiting them. Financial support from the Young Scientists Fund of the National Natural Science Foundation of China (NSFC) grant No.72203077 is gratefully acknowledged.}}

\author{
Hiroaki Kaido \\  Boston University \\ \url{hkaido@bu.edu}  
\and
Yi Zhang \\   Jinan University \\ \url{yzhangjnu@outlook.com} 
}

\maketitle

\begin{abstract}
A growing number of empirical models exhibit set-valued predictions. 
This paper develops a tractable inference method with finite-sample validity for such models. The proposed procedure uses a robust version of the universal inference framework by \cite{Wasserman:2020aa} and avoids using moment selection tuning parameters, resampling, or simulations. The method is designed for 
constructing confidence intervals for counterfactual objects and other functionals of the underlying parameter. It can be used in applications that involve model incompleteness,  discrete and continuous covariates, and parameters containing nuisance components.

\vspace{0.3in}
\end{abstract}

\clearpage
\section{Introduction}

\onehalfspacing
Discrete choice models are widely used to study economic agents' decisions. Such models may make \emph{set-valued predictions} when the researcher is willing to relax restrictive assumptions. A commonly used structure is that given observable and unobservable exogenous variables $(X,U)$, multiple values $G(U|X;\theta)$ of an outcome variable $Y$ are predicted. However, how $Y$ is selected from the predicted set is unspecified.
 A growing number of empirical models exhibit such predictions. Examples include but are not limited to discrete games \citep{bresnahan1990entry,Ciliberto:2009aa}, dynamic panel data models \citep{Honore:2006aa,chesher2024robust},
discrete choice models with heterogeneous choice sets \citep{bar:cou:mol:tei21}, models with instrumental variables \citep{che:ros17,Berry:2022aa}, auctions \citep{Haile:2003to}, network formation \citep{Miyauchi:2016aa,Sheng2020}, product offerings \citep{EIZENBERG:2014aa}, exporter's decisions \citep{Dickstein:2018aa} and school choices \citep{FackGrenetHe17}.
This class nests classic discrete choice models such as binary, multinomial, and ordered choice models, in which $G$ contains a single value of $Y$. 
Despite the recent developments in the econometrics literature, the existing methods for estimating such \emph{incomplete models} often face challenges in applications.  In particular, the limiting distributions of the existing statistics for full-vector and sub-vector inference are often non-standard and require complex regularity conditions or tuning parameters.

We propose a novel inference method for conducting hypothesis tests and constructing confidence sets (or intervals) in incomplete models.
Specifically, we consider testing the composite hypotheses:
\begin{align}
H_0:\theta\in\Theta_0~~~\text{v.s.}~~~H_0:\theta\in\Theta_1~,
\end{align}
for subsets $\Theta_0$ and $\Theta_1$ of a parameter space $\Theta$.  Our proposal is to compare a tailor-made likelihood-ratio statistic to a \emph{fixed} critical value, which builds on the \emph{universal inference} method of \cite{Wasserman:2020aa}. We show that the test is universally valid, meaning the proposed test has a correct size in any finite sample without complex regularity conditions.

Since the critical value is fixed, the proposed method remains computationally tractable in models of practical relevance. For example, we allow models to involve a number of discrete or continuous covariates and their associated parameters. %\footnote{Our empirical application involves about 6,000 observations and 10-15 covariates, including both discrete and continuous variables.}
The proposed test also avoids the use of moment selection or other regularization tuning parameters.
Inverting the test yields confidence sets (or intervals) for the structural parameter, their subcomponents, and counterfactual objects.

This paper aims to develop a versatile inference procedure that is widely applicable. 
We demonstrate that any incomplete discrete choice model possesses a structure that allows us to conduct the proposed test.
The main idea is as follows: When trying to distinguish between the null parameter set $\Theta_0$ and any (unrestricted) parameter value $\theta_1\in \Theta$, a classic idea is to use a likelihood ratio (LR) to compare the model fit at each $\theta\in\Theta_0$ against $\theta_1$. While the model incompleteness makes it difficult to apply this idea to the current setting, we show that a special pair of densities called \emph{least-favorable pair} exists for testing $\theta\in\Theta_0$ against $\theta_1$. This insight allows us to construct a parametric model $\{q_\theta\}$ indexed only by $\theta$, which is useful for constructing a robust test.\footnote{This parametric model can be derived analytically in many existing models \citep[see][]{kaido2019robust}. We also provide a command to implement the procedure as part of a general Python library for incomplete discrete choice models (\url{https://github.com/hkaido0718/IncompleteDiscreteChoice}).}
Using this construction, we obtain a novel LR statistic for universal inference.

Our framework allows the researcher to leave parts of individuals' decision processes unspecified and heterogeneous.
For example, \cite{bar:cou:mol:tei21} left the unknown choice-set formation process flexible.  The choice of an insurance plan may involve some buyers forming their choice set using an online platform, while others may form their choice sets based on recommendations by their insurance agents. These potentially heterogeneous choice set formation processes are unobservable to the econometrician and may not be understood well enough to be modeled.\footnote{If the researcher understands how the process operates, they can formulate it using a probabilistic model. Some studies take this approach \citep{bjorn_vuong_1984,Bajari:2010aa}.}
In an extreme case, each decision maker may have their own choice-set formation process, rendering them \emph{incidental parameters}.
While some empirical studies impose regularity conditions that implicitly limit the heterogeneity of decision processes, our proposed procedure remains valid without such assumptions.\footnote{Existing studies often impose some assumption on the heterogeneity of outcome distributions to apply limit theorems.  A leading example is to assume identically distributed outcomes, which holds under identically distributed selection mechanisms. A weaker assumption is that samples are stationary and strongly mixing \citep{chernozhukov2007estimation,andrews2010inference}, which imposes implicit restrictions on the selection mechanisms.
An exception is \cite{Epstein:2016qv}, which allowed arbitrary heterogeneity and dependence.}

The proposed procedure has two novel features compared to standard treatments of LR tests. Integrating them into the standard inference toolkit is straightforward. 
First, it utilizes a tailor-made likelihood function that can be derived by solving a convex program. We provide an algorithm and demonstrate that leading examples have closed-form likelihoods. Second, the test employs cross-fitting. Specifically, we obtain a parameter estimate from one subsample, use the other subsample to evaluate the likelihood ratio, and aggregate the statistics after swapping the roles of the two subsamples. This approach enables us to apply a straightforward conditioning argument and a Chernoff-style bound developed in \cite{Wasserman:2020aa}, yielding a simple yet non-trivial critical value. 
  
The method proposed here offers a tractable way for practitioners to perform inference on various models while avoiding ad-hoc assumptions. It is particularly effective in settings where the outcome is relatively low-dimensional, involves discrete and continuously distributed covariates, and $\theta$ may contain nuisance parameters. Such examples are ubiquitous. 
  For the theoretical front, this paper provides a new likelihood-based inference method with finite sample validity with a focus on inference for subvectors and counterfactual objects. 
  To keep a tight focus, we limit our attention to the finite-sample validity of the test and evaluate its power only through simulations. The proposed procedure is not designed to be optimal due to the use of sample-splitting and a Chernoff-style bound. In a companion paper \citep{kaido2019robust}, we build a theoretical framework for analyzing the power of likelihood-ratio tests and their optimality.

\subsection{Relation to the literature}
The study of incomplete systems has a long history dating back to the work of \cite{Wald1950}. Economic models with multiple equilibria are well-known examples of incomplete models. \cite{jovanovic1989observable} developed a theoretical framework to examine the empirical content of such models. The class of incomplete models considered here also covers a wide range of empirical models beyond them, as discussed earlier. 
Systematic ways to derive partially identifying restrictions for such models have been developed \citep{tamer2003incomplete,galichon2011set,beresteanu2011sharp,che:ros17}. Building on this line of work, we use the \emph{sharp identifying restrictions} to incorporate all information in the original structural model to construct a likelihood function.

We contribute to the literature on inference by providing a novel test that has the following properties: (i) it has finite sample validity without complex regularity conditions (universal validity); (ii) it is applicable to models with mixed data types (e.g., continuous and discrete covariates); (iii) one can construct confidence regions for the entire parameter or its functions;  (iv) it is robust to the presence of incidental parameters (selections); (v) it can accommodate nonparametric parameter components.\footnote{Each of the properties has been studied somewhat separately. For (i), other than \cite{Li:2022aa},  \cite{Horowitz:2023aa} develop a method with finite sample validity for models represented by optimization problems; (ii) For discrete covariates, one can use unconditional moment inequalities \citep[see][and references therein]{Canay_Shaikh_2017}. For continuous covariates, one needs to work with a continuum (or increasing number) of moment inequalities \citep{Andrews:2013aa,chernozhukov_lee_rosen13}. See also \cite{Kaido:2023aa} for this point who develop a method related to this paper's approach;   (iii) Subvector inference is studied, for example, by \cite{BCS,KMS,Cox:2022aa,Andrews:2023aa}; For (iv), \cite{Hahn:2010aa} tackles the problem using panel data. \cite{Epstein:2015aa,Epstein:2016qv} develop robust Bayesian and frequentist inference methods; For (v), there is a vast literature on semiparametric models \citep{POWELL19942443}.}
To our knowledge, \cite{Li:2022aa} is the only existing work that develops a finite-sample valid method for incomplete models using the idea of Monte Carlo tests.
Our proposal differs from theirs mainly in two respects.
First, we focus on composite hypothesis tests with inference on subvectors and counterfactual objects in mind, while they focus on full-vector inference. Second, we use a fixed critical value, which avoids simulated draws of latent variables used in their method to calculate critical values. Our test uses a likelihood-ratio statistic. Likelihood-based methods are also considered by \cite{Chen:2011aa,Chen_2018} and \cite{Kaido:2023aa} who show their methods' asymptotic uniform validity. In contrast, this paper aims to achieve finite-sample validity.

Likelihood-based inference is commonly used. However, as \cite{Wasserman:2020aa}  notes ``The (limiting) null distribution of the classical likelihood-ratio statistic is often intractable when used to test composite null hypotheses in irregular statistical models.''
 They develop inference methods that do not require complex regularity conditions using split-sample and cross-fit versions of LR statistics. This approach is also appealing for incomplete models, which are typically irregular.
 Nonetheless, \citeposs{Wasserman:2020aa} framework is not directly applicable to the current setting because their method uses a unique likelihood function as a key input and requires the knowledge of the sampling process.\footnote{As a baseline, they assume random sampling. To be precise, they assume the knowledge of the conditional likelihood based on a subsample given another subsample, which is unknown in incomplete models. See Section \ref{ssec:overview}.} Neither of them is readily available in incomplete models. This is because, for each parameter, the model implies multiple (typically infinitely many) likelihoods. 
Furthermore, the unspecified selection mechanism (together with observable and unobservable variables) can vary across experiments. We address these issues by carefully constructing a likelihood function using the model structure, hence the name ``tailor-made'' likelihood.

The universal validity is related to the \emph{asymptotic uniform validity} requirements shown for many of the existing proposals based on moment inequality restrictions \citep[see][for a thorough review]{Canay_Shaikh_2017}. Both are the notions of the validity of an inference procedure over a wide class of data-generating processes. The former requires validity in any finite samples, whereas the latter requires it in large samples.  \cite{Wasserman:2020aa} also emphasize that they avoid regularity conditions in defining the set of data-generating processes, which we also follow here.
Some of the recently developed moment-based inference methods study subvector inference and attain computational tractability by focusing on specific classes of models or testing problems \citep{Cox:2022aa,Andrews:2023aa}. This paper focuses on the class of incomplete discrete choice models, which is not nested by (nor nests) the class of models considered in these proposals.

\section{Set-up}
Let $Y\in \cY\subseteq\mathbb R^{d_Y}$ and $X\in \cX\subseteq\mathbb R^{d_X}$ denote, respectively, observable endogenous and exogenous variables, and $U\in \cU\subseteq\mathbb R^{d_U}$ denote latent variables. We assume $\cY$ is a finite set.  Let $\cS=\cY\times\cX$. We equip $\cS$ with the Borel $\sigma$-algebra $\Sigma_{\cS}$ and let $\Delta(\cS)$ denote the set of all Borel probability measures on $(\cS,\Sigma_\cS)$.
Let $\Theta$ be a parameter space. We do not restrict $\Theta$. Hence, both parametric and nonparametric components are accommodated. Let $F_\theta(\cdot|x)$ denote the conditional distributions of $U$ given $X=x$. We let $F=\{F_\theta,\theta\in\Theta\}$ be the collection of the conditional laws.

 For each $\theta\in\Theta$, let 
 $G(\cdot|\cdot;\theta):\cU\times\cX\twoheadrightarrow \cY$ be a weakly measurable correspondence, which collects permissible outcome values for each $(x,u)$. 
The observable outcome $Y$ is a random vector satisfying
\begin{align}
    Y\in G(U|X;\theta),~ a.s. \label{eq:correspondence}
\end{align}
Such a random vector is called a \emph{measurable selection} of $G(U|X;\theta)$.
The model does not impose any restrictions on how $Y$ is selected. This structure nests models with a \emph{complete prediction}, characterized by a function $g(\cdot|\cdot;\theta):\cX\times\cU\to\cY$ such that $Y=g(U|X;\theta),~a.s.$

Let $\cC$ be the collection of all subsets of $\cY$.
Define the \emph{containtment functional} of $G$ by
\begin{align}
	\nu_\theta(A|x)\equiv \int 1\{G(u|x;\theta)\subseteq A\}dF_\theta(u|x),~A\in\cC.\label{eq:defnu}
\end{align} 
This functional characterizes \emph{all} conditional distributions of the measurable selections of $G(U|x;\theta)$ by the following set, known as the \emph{core} of $\nu_\theta$ \citep[Theorem 2.1]{art83}:
\begin{align}
	\mathcal P_{\theta,x}\equiv\{Q\in\mathcal M(\Sigma_Y,\cX):Q(A|x)\ge \nu_\theta(A|x),~A\in \cC\},\label{eq:artstein}
\end{align}
where $\cM(\SY,\cX)$ is the collection of laws of random variables supported on $\cY$ conditional on $\ex$. The inequalities $Q(\cdot|x)\ge \contf(\cdot|x)$ characterizing $\mathcal P_{\theta,x}$ are known as the \emph{sharp identifying restrictions} as they can partially identify $\theta$ without losing information \citep[see e.g.,][]{galichon2011set,mol:mol18}. Our method applies to any model characterized by such restrictions. 

Any (conditional) distribution in $\mathcal P_{\theta,x}$ can also be expressed as follows \citep{philippe1999decision}:
\begin{align}
	Q(\cdot|x)=	\int_{\cU}\eta(\cdot|u,x) dF_\theta(u|x),~	\eta(\cdot|u,x)\in \Delta(G(u|x;\theta)).\label{eq:selection_rep}
\end{align}
The unknown function $\eta(\cdot|u,x)$ is the conditional distribution of $Y$ over the prediction, representing an unknown \emph{selection mechanism}. It may represent objects with different interpretations, such as equilibrium selection mechanisms (Example \ref{ex:game1}),  choice-set formation processes (Example \ref{ex:choice_set}), fixed effects, and the initial conditions (Example \ref{ex:ddc}).\footnote{Other examples include the unobserved true value of variables in a set-valued covariate or control function  \citep{manskitamer02,han2024setvalued} and details of agents' behavior beyond minimal rationality or stability restrictions \citep{Haile:2003to,FackGrenetHe17}.}
The characterization in \eqref{eq:artstein} is tractable as it allows us to restrict the conditional distribution of $Y$ only through inequality restrictions without introducing $\eta(\cdot|u,x)$ explicitly.

Let $\mu$ be the counting measure on $\cY.$
For each $x\in\cX$, let
\begin{align}
	\fq{\theta,x}\equiv\{q:q(\cdot|x)=dQ(\cdot|x)/d\mu,~Q(\cdot|x)\in \mathcal P_{\theta,x}\}\label{eq:deffq}
\end{align}
collect the set of conditional densities compatible with $\theta$ at $x$. We then let $\fq{\theta}\equiv\{\fq{\theta,x},x\in \cX\}$.

\subsection{Examples}

We illustrate $G$ with well-known examples. 
\begin{example}[Discrete Game]\label{ex:game1}
	Consider a two-player static game of complete information \citep{bresnahan1990entry,bresnahan1991empirical,tamer2003incomplete}. Each player may either choose $y^{(j)}=0$ or $y^{(j)}=1$. Let $x^{(j)}$ and $u^{(j)}$ be player $j$'s observable and unobservable characteristics. 
 The payoff of player $j$ is 
	\begin{align}
	\pi^{(j)}=y^{(j)}\big(x^{(j)}{}{'}\delta^{(j)}+\beta^{(j)}y^{(-j)}+u^{(j)}\big),~j=1,2\label{eq:payoff}
	\end{align}
	where $y^{(-j)}\in\{0,1\}$ is the opponent's action. Let $\theta = (\beta', \delta')'$, and assume the effect of the entry of the opponent is nonpositive, i.e., $\beta^{(j)}\le 0,\forall j$. Then, the following correspondence gives the set of pure strategy Nash equilibria (PSNE):
\citep[Proposition 3.1]{beresteanu2011sharp}:
\begin{align}
G(u|x;\theta)=
\begin{cases}
\{(0, 0)\} & u\in S_{\{(0,0)\}|x;\theta}\equiv \{u:u^{(j)}<-x^{(j)}{}'\delta^{(j)},j=1,2\},\\
\{(0, 1)\} & u\in S_{\{(0,1)\}|x;\theta}\equiv \{u^{(1)}<-x^{(1)}{}'\delta^{(1)}, u^{(2)}>-x^{(2)}{}'\delta^{(2)}\}\\
&\qquad \cup\{-x^{(1)}{}'\delta^{(1)}<u^{(1)}<-x^{(1)}{}'\delta^{(1)}-\beta^{(1)}, u^{(2)}>x^{(2)}{}'\delta^{(2)}-\beta^{(2)}\},\\
\{(1, 0)\} & u\in S_{\{(1,0)\}|x;\theta}\equiv \{u^{(1)}>-x^{(1)}{}'\delta^{(1)}-\beta^{(1)}, u^{(2)}<-x^{(2)}{}'\delta^{(2)}-\beta^{(2)}\}\\
&\qquad \cup\{-x^{(1)}{}'\delta^{(1)}<u^{(1)}<-x^{(1)}{}'\delta^{(1)}-\beta^{(1)}, u^{(2)}<-x^{(2)}{}'\delta^{(2)}\},\\
\{(1, 1)\} & u\in S_{\{(1,1)\}|x;\theta}\equiv \{u:u^{(j)}>-x^{(j)}{}'\delta^{(j)}-\beta^{(j)},j=1,2\},\\
\{(1, 0), (0, 1)\} & u\in S_{\{(0,1),(1,0)\}|x;\theta}\equiv \{u:-x^{(j)}{}'\delta^{(j)}<u^{(j)}<-x^{(j)}{}'\delta^{(j)}-\beta^{(j)}, j=1, 2\}.
\end{cases}\label{eq:strategic_sub}
\end{align}
The model admits multiple equilibria when $u\in S_{\{(0,1),(1,0)\}|x;\theta}$ (see Figure \ref{fig:G_entry}). Below, we use this model as a working example. 

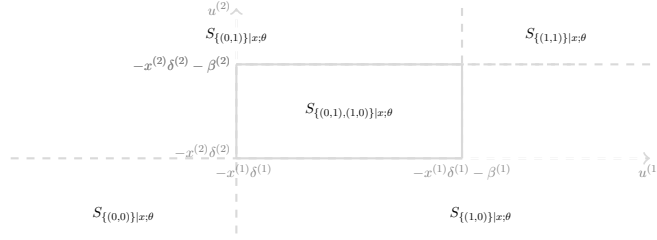
\begin{figure}[htbp]
\begin{center}
\begin{tikzpicture}[thick,scale=0.5, every node/.style={transform shape}]
		\draw[->,dashed,gray!30] (-4,0) -- (13,0) node[below] {$\color{gray}u^{(1)}$};
		\draw[->,dashed,gray!30] (2,-2) -- (2,4) node[left] {$\color{gray}u^{(2)}$};

		\draw[white] (2,2.5) -- (2,3.7);
		\draw[white] (8,0) -- (12.7,0);

		\draw (2,3.25) node {$S_{\{(0,1)\}|x;\theta}$};

		\draw (8.5,-1.5) node {$S_{\{(1,0)\}|x;\theta}$};

		\draw (5,1.25) node {$S_{\{(0,1),(1,0)\}|x;\theta}$};

		\draw[dashed,gray!30] (13,2.5) -- (2,2.5) node[left] {$\color{gray}-x^{(2)}\delta^{(2)}-\beta^{(2)}$};
		\draw[dashed,gray!30] (8,4) -- (8,0) node[below] {$\color{gray}-x^{(1)}\delta^{(1)}-\beta^{(1)}$};
		\draw[dashed,gray!30] (11,2.5) -- (2,2.5) node[left] {$\color{gray}-x^{(2)}\delta^{(2)}-\beta^{(2)}$};

		\draw[gray!30] (2,0) -- (2,2.5);
		\draw[gray!30] (8,0) -- (8,2.5);
		\draw[gray!30] (2,0) -- (8,0);
		\draw[gray!30] (2,2.5) -- (8,2.5);

		\draw (10.5,3.25) node {$S_{\{(1,1)\}|x;\theta}$};

		\draw (-1,-1.5) node {$S_{\{(0,0)\}|x;\theta}$};

		\draw[gray!30] (2,0.2) node[left] {$\color{gray}-x^{(2)}\delta^{(2)}$};
		\draw[gray!30] (2.2,0) node[below] {$\color{gray}-x^{(1)}\delta^{(1)}$};		
		\end{tikzpicture}
\caption{\footnotesize{Level sets of $\G(\cdot|x;\theta)$ with $\beta^{(j)}< 0,j=1,2$. }}
		\label{fig:G_entry}
		\end{center}
\end{figure}
\end{example}

The next example concerns a single-agent discrete choice problem with unobserved choice sets.
\begin{example}[Heterogeneous Choice Sets]\label{ex:choice_set}
Let $\mathcal{J} = \{1,\dots,J\}$.  An individual draws a choice set $C\subseteq\mathcal{J}$ and  chooses the alternative $\ey\in C$ that maximizes her utility:
\begin{align*}
\ey=\argmax_{j\in C} \left(\pi(\ex_{j};\theta)+\eu_j\right).
\end{align*}    
Each alternative is characterized by a vector of observable covariates $\ex_j$ and a latent variable. Suppose $|C|\ge \kappa$ with probability 1 for some known $\kappa\ge 2.$ Without further assumptions, \citet[Lemma A.1]{bar:cou:mol:tei21} show that the set of optimal choices is a measurable correspondence:
\begin{align*}
\G(u|x;\theta)=\cup_{K\subseteq\mathcal{J}:|K|=\kappa}\left\{\arg\max_{j \in K}  \left(\pi(x_{j};\theta)+u_j\right)\right\}.
\end{align*} 
\end{example}

The next example is a panel dynamic discrete choice model with short panel data \citep{Honore:2006aa,KHAN2023105515,chesher2024robust}.
\begin{example}[Panel Binary Choice Model]\label{ex:ddc}
Suppose a binary outcome for individual $i$ in period $t$ is generated according to
	\begin{align}
	    Y_{it}&=1\{g(X_{it};\theta)+\gamma Y_{it-1}+A_i+U_{it}\ge 0\},~t=1,\dots,T,
	\end{align}
 for some function $g$ known up to $\theta\in\Theta$. 
The fixed effect $A_i$ is allowed to be correlated with $(X_i,U_i)\equiv(X_{it},U_{it})_{t=1}^T$ arbitrarily. The initial outcome $Y_{i0}$ is unobserved.
Hence, the analyst is does not know how $(A_i, Y_{i0})$ is distributed conditional on $(X_{i},U_{i})$. Let $x=(x_1,\dots,x_T)'$ and $u=(u_1,\dots,u_T)'.$
The model's prediction is as follows.
\begin{multline}
G(u|x;\theta)=\Big\{y=(y_1,\dots,y_T):y_{t}=1\{g(x_{t};\theta)+\gamma {y_{t-1}}+a+u_{t}\ge 0\}, ~t=1,\dots, T, \\
\text{for some }(y_{0},a)\in \{0,1\}\times \mathbb R\Big\}.\label{eq:panel1}
\end{multline}
This model contains empirical contents about an individual's switching decisions from one alternative to the other \citep{chesher2024robust}.\footnote{A static version of this model is closely related to \citeposs{Manski:1987aa} model, in which $g(x_t;\theta)=x_t'\theta$. He showed identification of $\theta$ (up to scale) is possible if the marginal distribution of $U_{it}|X_i$ is time invariant. The model is also related to \citeposs{Chamberlain:2010aa} model, which additionally assumes $A_i$ is independent of $(X_i,U_i)$. }
\end{example}

\subsection{An Overview of the Main Results}\label{ssec:overview}
This section overviews the proposed procedure.
Let $(Y_i,X_i),i=1,\dots, n$ be a sample of outcome and covariates.
Consider testing 
\begin{align}
H_0:\theta\in\Theta_0,~~~\text{v.s.}~~~H_0:\theta\in\Theta_1~,
\end{align}
for subsets $\Theta_0,\Theta_1$ of the parameter space.
The proposed procedure takes the following steps.

\textbf{Step 1:} Split samples into $D_0$ and $D_1$.  Let $\hat\theta_1$ be any estimator of $\theta$ computed from sample $D_1$. 

\textbf{Step 2:} Using $D_0$, construct a tailor-made likelihood function 
\begin{align}
	\cL_0(\theta)&=\prod_{i\in D_0} q_{\theta}(Y_i|X_i),~\theta\in \Theta_0\cup\{\hat\theta_1\},\label{eq:def_cL}
\end{align}
where $q_\theta(y|x)$ is the \emph{LFP-based density} we introduce below.
Let $\hat\theta_0$ be the \emph{restricted maximum likelihood estimator (RMLE)} based on $D_0$:
\begin{align}
	\hat\theta_0\in \argmax_{\theta\in\Theta_0} \cL_0(\theta);\label{eq:theta0}
\end{align}

\textbf{Step 3:} 
Compute the \emph{split-sample likelihood-ratio (LR) statistic}:
\begin{align}
	T_n=\frac{\cL_0(\hat\theta_1)}{\cL_0(\hat\theta_0)}.\label{eq:def_Tn}
\end{align}
The \emph{cross-fit LR statistic} is
\begin{align}
S_n=\frac{T_n+T_n^{\text{swap}}}{2}	,\label{eq:def_Sn}
\end{align}
where $T_n^{\text{swap}}$ is calculated in the same way as $T_n$ after swapping the roles of $D_0$ and $D_1$.

\textbf{Step 4:} Reject $H_0$ if $S_n> 1/\alpha$. Do not reject $H_0$ otherwise. 

One can construct confidence regions for $\theta$ and its functions by defining $\Theta_0$ properly and inverting the test (see Section \ref{sec:main}).

We focus on key properties of the proposed test for now and defer discussion on how to construct $\hat\theta_1$ and $\cL_0$ to the next sections.
The cross-fit LR test controls its size in any finite sample. That is, for any $n$ and the class $\cP^n_0$ of data-generating processes (DGPs) compatible with the null hypothesis, the following statement holds 
\begin{align}
	\sup_{P^n\in \cP^n_0}P^n\big(S_n>\frac{1}{\alpha}\big)\le \alpha.
\end{align}
Our test builds on \cite{Wasserman:2020aa}, who introduced the split-sample and cross-fit LR statistics for probabilistic models. Following them, we call the procedure above \emph{universal inference} (or \emph{universal hypothesis test}). The idea is that the inference applies universally to any model described by \eqref{eq:correspondence} without further regularity conditions.

In the next few sections, we provide details on how to construct the LR statistic.  The key is to construct $\cL_0$ from a \emph{least favorable} parametric model $\{q_\theta,\theta\in\Theta_0\}$ against a density $q_{\hat\theta_1}$ in the unrestricted model.

\subsection{Unrestricted Estimator}\label{ssec:theta1}
Any estimator of $\theta$ based on sample $D_1$ can be used as an unrestricted estimator $\hat\theta_1$. The main role of this estimator is to find a ``representative'' parameter value that fits $D_1$ well. Hence, we do not require $\hat\theta_1$ to be a consistent estimator. Nevertheless, we recommend choosing $\hat\theta_1$ to maximize a measure of fit to $D_1$ for power consideration. 

A natural choice is an extremum estimator that minimizes some sample criterion function $\theta\mapsto\hat {\mathsf Q}_1(\theta)$.\footnote{We use the subscript ``1'' to indicate that the sample criterion function onlyuses $D_1$.} An example of $\hat {\mathsf Q}_1$ is 
\begin{align}
    \hat{\mathsf Q}_1(\theta)=\sup_{j,x}\frac{\{\nu_\theta(A_j|x)-\hat P_1(A_j|x)\}_+}{\hat{s}_{\theta,1}(A_j|x)},\label{eq:cht}
\end{align}
where
$\hat P_1(A_j|x)$ is an estimator of the conditional probability $P(A_j|x)$ using $D_1$, and $\hat{s}_{\theta,1}$ is an estimator of the standard error of $\hat P_1$.\footnote{These criterion functions are commonly used in practice \citep{chernozhukov2007estimation,chernozhukov_lee_rosen13}. The supremum operation over $j$ (or $x$) can be replaced with sum (or integral). See \cite{Andrews:2013aa,chernozhukov_lee_rosen13}.}  Another possibility is to use
the (negative) log-likelihood function:
\begin{align}
    \hat{\mathsf Q}_1(\theta)=\sum_{i\in D_1}-\ln p_\theta(Y_i|X_i;\hat p_n).\label{eq:klic_proj}
\end{align}
It uses a log-likelihood $p_\theta(\cdot|\cdot;\hat p_n)$, which is the Kullback-Leibler divergence projection of a nonparametric estimator of the conditional choice probability to $\fq{\theta,x}$ \citep{Kaido:2023aa}. Our simulation results suggest both estimators perform similarly.

Next, we find a positive density $p(\cdot|x)$ compatible with the unrestricted estimator, by solving the following linear feasibility problem: 
\begin{align}
	\text{Find }&p(\cdot|x)\in \Delta^{\cY}\label{eq:defp}\\
	s.t.~&\sum_{y\in A}p(y|x)\ge \nu_{\hat\theta_1}(A|x),~ A\in \cC.\notag
\end{align}
Any solution can be used as long as $p(y|x)>0$ for all $y\in\cY$.
We then set $q_{\hat\theta_1}=p$. One may view $q_{\hat\theta_1}$ as a representative density in the unrestricted model.

\begin{remark}
We use $\hat\theta_1$ to compute the unrestricted density $q_{\hat\theta_1}$ in order to attain nontrivial power. When $H_1$ is true, $\cL_0(\hat\theta_1)$ is expected to be away from the maximal value of $\cL_0(\cdot)$ over $\Theta_0$, reflecting the classic idea behind LR tests. This is not the only way to construct an unrestricted density. 
Another approach is to use the idea of weighted average power \citep{Andrews:1994aa,Elliott:2015aa}. The companion paper takes this approach with a full-sample likelihood. This approach requires specifying a prior (over $\Theta$) and calculating an integral of the containment functional. To keep the method simple, this paper does not take this route.
\end{remark}

\subsection{How to Construct $\cL_0$}\label{ssec:cL0}
Below, we condition on $D_1$, assume $q_{\hat\theta_1}$ is already computed, and treat it as fixed. 
Given $q_{\hat\theta_1}$, consider testing whether the data in $D_0$ are compatible with the null hypothesis or $q_{\hat\theta_1}$. If there is a parametric model $\{q_\theta,\theta\in\Theta_0\}$ such that $q_\theta\in\fq{\theta,x}, \theta\in\Theta_0$,
it is natural to base our test on:
\begin{align}
 \frac{\prod_{i\in D_0}q_{\hat\theta_1}(Y_i|X_i)}{\sup_{\theta\in\Theta_0} \prod_{i\in D_0}q_\theta(Y_i|X_i)}.\label{eq:heuristic_LR}
\end{align} 
However, naively selecting a parametric model can be problematic.
For each $\theta$, some density in $\fq{\theta,x}$ may be easier to distinguish from $q_{\hat\theta_1}$ than others. Using such densities to compute the denominator of \eqref{eq:heuristic_LR} could cause over-rejection if the true density is close to $q_{\hat\theta_1}$.

To address this issue, we define the \emph{least favorable pair (LFP)-based parametric model}. 
For this, define the Kullback-Leibler (KL) divergence by
\begin{align}
	I(f(\cdot|x)||f'(\cdot|x))\equiv \int_{S_x}\ln \frac{f(y|x)}{f'(y|x)}f(y|x)d\mu,
\end{align}
where $S_x=\{y\in\cY:f(y|x)>0\}$.

\begin{definition}[LFP-based parametric model]\label{def:lfp_model}
A family of densities $\{q_\theta,\theta\in\Theta_0\cup\{\hat\theta_1\}\}$ is a \emph{least favorable pair (LFP)-based parametric model} if (i) for each $\theta\in \Theta_0$ and $x\in\cX$,
\begin{align}
	q_\theta(\cdot|x)=\argmin_{q(\cdot|x)\in \fq{\theta,x}}&~I(q(\cdot|x)+p(\cdot|x)||q(\cdot|x))~,\label{eq:klsol}
\end{align}
for $p(\cdot|x)\in \fq{\hat\theta_1,x}$; and (ii) $q_{\hat\theta_1}(\cdot|x)=p(\cdot|x)$.
\end{definition}
We call $q_\theta$ the \emph{LFP-based density}.
The theory of minimax tests ensures that $q_\theta(\cdot|x)$ is the least-favorable density for distinguishing $\fq{\theta,x}$ from $q_{\hat\theta_1}(\cdot|x)$ (see Proposition \ref{prop:cs}), suitable for constructing an LR statistic and achieving size control. A practical aspect of this construction is that one can compute $q_\theta(\cdot|x)$ by solving the following convex program:
\begin{align}
q_\theta(\cdot|x)=\argmin_{q(\cdot|x)\in \Delta^\cY}&~\sum_{y\in \cY}\ln\Big(\frac{q(y|x)+p(y|x)}{q(y|x)}\Big)(q(y|x)+p(y|x)),\label{eq:def_qtheta}\\
	s.t.&~\sum_{y\in A}q(y|x)\ge  \nu_{\theta}(A|x),~A\in \cC, x\in X.\notag
\end{align}
The program \eqref{eq:def_qtheta} imposes the sharp identifying restrictions as constraints. As such, they retain all restrictions from the model. 
We show that $q_\theta$ can be derived analytically in leading examples.
A companion Python program can also be used to compute $q_\theta$ numerically.
\begin{remark}\label{rem:cardinality}
The convex program \eqref{eq:def_qtheta} (or \eqref{eq:defp}) is computationally tractable as long as the number of inequalities is moderate. This is the case when the cardinality of $\cY$ is not too high, e.g., games with a small number of players or dynamic models with short panel data. In richer models, several strategies are recommended. One approach is to reduce the number of inequalities to a manageable size without losing information. The recent work by \cite{luo_ponomarev_wan} demonstrates this can be done for various models and provides a graph-based algorithm to obtain such inequalities (called the smallest core-determining class (CDC)).\footnote{They also show that the smallest CDC only depends on the support of $G(U|X;\theta)$. Hence, as long as the support remains the same in a given model, the smallest CDC needs to be computed only once.} The second approach is to add a restriction to the model. For example, in discrete games, grouping players into several types helps reduce the number of inequalities \citep{Berry_Tamer_2006}. Finally, another possibility is to use a non-sharp subset of inequalities (see Remark \ref{rem:approximation}).  
\end{remark}

Now we define our tailor-made likelihood function and the restricted maximum likelihood estimator as follows:\footnote{To compute $T_n$, we only need $q_\theta$ to be defined over $\Theta_0\cup \{\hat\theta_1\}$, which is a non-random domain conditional on $D_1$. The LFP-based parametric model is defined similarly for $T_n^{\text{swap}}.$}
\begin{align}
\cL_0(\theta)\equiv \prod_{i\in D_0}q_\theta(Y_i|X_i), 	~\theta\in\Theta_0\cup \{\hat\theta_1\},\qquad \hat\theta_0\in \argmax_{\theta\in \Theta_0}\cL_0(\theta).
\end{align}

\section{Finite Sample Validity}\label{sec:main}

We make the following assumption on sampling. 
\begin{assumption}\label{as:iidF}
(i) $(Y_i,X_i,U_i),i=1,\dots,n$ are independently distributed across $i$; (ii) $(X_i,U_i),i=1,\dots,n$ are identically distributed.
\end{assumption}
We assume $(X_i,U_i)$ are independently and identically distributed across $i$.\footnote{Allowing heterogeneity of $F_{U|X}$ across $i$ is straightforward but requires additional notation.} 
We also assume $Y_i$ is independently distributed but allow its law to be heterogeneous across $i$.  Allowing $Y_i$'s (cross-sectional) dependence is also possible at the expense of modifying $\hat\theta_1$. We discuss this extension in Appendix \ref{sec:dependence}.

For each $\theta\in\Theta$, let
\begin{align}
	 \mathcal P^n_\theta\equiv\Big\{P^n\in\Delta(\cS^n):P^n=\bigotimes_{i=1}^n P_i, P_i(\cdot|x)\in\cP_{\theta,x},~	\forall i,~ x\in\cX \Big\}.\label{eq:cPtheta}
\end{align}
Let $\cP^n_0\equiv\{P^n\in\cP^n_\theta:\theta\in\Theta_0\}$ be the set of data generating processes (DGPs) compatible with $H_0$. The following theorem establishes the universal validity of the proposed test and its robustness.
\begin{theorem}\label{thm:univ_inf}
Suppose Assumption \ref{as:iidF} holds. Then, for any $n\in\mathbb N$,
\begin{align}
	\sup_{P^n\in \cP^n_0}P^n\big(S_n>\frac{1}{\alpha}\big)\le \alpha.
\end{align}
\end{theorem}
The theorem ensures that the test's finite-sample size is at most $\alpha$ under any distribution $P^n$ compatible with the null hypothesis regardless of how selections operate across $i$.
The result does not require any regularity conditions beyond Assumption \ref{as:iidF}. The proof is in Appendix \ref{ssec:proofs}. It uses the properties of the least-favorable pair and Markov's inequality interpreted as a Chernoff-style bound for log likelihood-ratio as in \cite{Wasserman:2020aa}. For our result, the construction of the null model $\{q_\theta,\theta\in\Theta_0\}$ is the key, which is possible because  $\fq{\theta,x}$ is characterized by a containment functional. We elaborate on this theoretical background in Section \ref{sec:background}. 

Let us compare our test with the standard LR test. They are based on the idea of comparing the likelihood values with and without restrictions imposed by $H_0$.  The standard LR test's rejection rule is $2\ln\big(\cL_{full}(\hat\theta_1)/\cL_{full}(\hat\theta_0)\big)> c_{df,\alpha}$, where $\cL_{full}$ is the full-sample likelihood, and the critical value $c_{df,\alpha}$ is the $1-\alpha$ quantile of a $\chi^2$-distribution, which increases as the degrees of freedom ($df$) increases. In contrast, a split-sample version of the proposed test uses $2\ln\left(\frac{\cL_0(\hat\theta_1)}{\cL_0(\hat\theta_0)}\right)> c_\alpha$ with $c_\alpha=2\ln(1/\alpha)$. This critical value does not increase with the degrees of freedom. This difference arises because we use the out-of-sample likelihood $\cL_0(\hat\theta_1)$ whose behavior differs from the in-sample likelihood $\cL_{full}(\hat\theta_1)$.\footnote{See \cite{Wasserman:2020aa} for further discussion on this point.}

Confidence regions for functions of $\theta$ can be constructed in a simple manner.  Let $\varphi:\Theta\to \mathbb R^{d_\varphi}$. Examples of $\varphi(\theta)$ are subvectors of $\theta$ and counterfactual objects as we discuss below using examples.\footnote{One can also construct a confidence region for the entire parameter by letting $\varphi$ be the identity map.}
For each $\varphi^*\in\mathbb R^{d_\varphi}$, let $\Theta_0(\varphi^*)\equiv\{\theta\in\Theta:\varphi(\theta)=\varphi^*\}$ and $\Theta_1(\varphi^*)\equiv\{\theta\in\Theta:\varphi(\theta)\ne\varphi^*\}$. 
Let $T_n(\varphi^*)$ be defined as in \eqref{eq:def_Tn} with $\Theta_0=\Theta_0(\varphi^*)$, and let $S_n(\varphi^*)$ be defined similarly. Define
\begin{align}
    CS_n\equiv\big\{\varphi^*\in\mathbb R^{d_\varphi}:S_n(\varphi^*)\le \frac{1}{\alpha}\big\}.\label{eq:ci}
\end{align}
For each $P^n$, let $\mathcal H_{P^n}[\varphi]\equiv\{\varphi^*:\varphi(\theta)=\varphi^*, P^n\in\cP^n_\theta,\text{ for some }\theta\in\Theta\}$ be the \emph{sharp identification region} for $\varphi(\theta)$. Let $\mathcal F^n\equiv \big\{(\varphi^*,P^n):\varphi(\theta)=\varphi^*,~P^n\in\cP^n_\theta,\text{ for some }\theta\in\Theta\big\}$ be the collection of data generating processes compatible with the model restrictions.
Then, the following result holds.
\begin{corollary}\label{cor:univ_cs}
Suppose Assumption \ref{as:iidF} holds. Then, for any $n\in\mathbb N$,
\begin{align}
	\inf_{(\varphi^*,P^n)\in \mathcal F^n}P^n\big(\varphi^*\in CS_n\big)\ge 1-\alpha.
\end{align}
\end{corollary}
The coverage statement can also be written as $\inf_{P^n\in \mathcal P^n_\theta,\theta\in\Theta}\inf_{\varphi^*\in \mathcal H_{P^n}[\varphi]}P^n(\varphi^*\in CS_n)\ge 1-\alpha$. Therefore, the corollary ensures that $CS_n$ covers elements of the sharp identification region uniformly across $P^n$. 

\begin{remark}
\cite{Wasserman:2020aa} (Theorem 5) demonstrate that, in regular models with point identified $\theta$\footnote{To be precise, they assume identifiability of $\theta$ and the differentiability of the statistical model in quadratic mean (DQM), together with availability of a $\sqrt n$-consistent estimator of $\theta$.}, their confidence set has diameter $O_p(\sqrt{\log(1/\alpha)/n})$, the same order as the optimal confidence set. In incomplete models, $\theta$ is typically partially identified. As such, it is not immediately clear whether such a result can be extended to the current setting. To keep a tight focus, we leave the analysis of (excess) diameters of $CS_n$ for future work.
\end{remark}

\subsection{Illustration}\label{sec:examples}
We illustrate the main result by revisiting Example \ref{ex:game1}.

\setcounter{example}{0}
\begin{example}[Discrete Game]
By \eqref{eq:defnu}-\eqref{eq:artstein}, \eqref{eq:deffq}, and \eqref{eq:strategic_sub}, 
the set of densities compatible with $\theta$ is
\begin{multline}
	\fq{\theta,x}=\Big\{q\in\Delta:~q((0,0)|x)=\f(S_{\{(0,0)\}|x;\theta}|x),~q((1,1)|x)=\f(S_{\{(1,1)\}|x;\theta}|x),\\
	  \f(S_{\{(1,0)\}|x;\theta}|x)\le q((1,0)|x)\leq \f(S_{\{(1,0)\}|x;\theta})+\f(S_{\{(0,1),(1,0)\}|x;\theta}|x)\Big\}.\label{eq:frak_q_CT}
\end{multline}
One can obtain $q_\theta$ in closed form by solving \eqref{eq:def_qtheta}. For this, let $\eta_1(x;\theta)\equiv 1-\f(S_{\{(0,0)\}|x;\theta}|x)-\f(S_{\{(1,1)\}|x;\theta}|x),
\eta_2(x;\theta)\equiv \f(S_{\{(1,0)\}|x;\theta}|x)+\f(S_{\{(0,1),(1,0)\};\theta}|x),$ and $
\eta_3(x;\theta)\equiv\f(S_{\{(1,0)\}|x;\theta}|x).$  $\eta_1$ is the probability allocated, under $\theta$, to either $(1,0)$ or $(0,1)$. Similarly, $\eta_2$ ($\eta_3$) is the upper (lower) bound, under $\theta$, on the probability of $(1,0)$ being the equilibrium outcome of the game.

The following proposition characterizes $q_\theta$.

 \begin{proposition}\label{prop:game_lfp}
Let $p\in \fq{\hat\theta_1}$, and let $p_{rel}((1,0)|x)\equiv\tfrac{p((1,0)|x)}{p((1,0)|x)+p((0,1)|x)}$. 
For any $\theta\in\Theta_0$, the LFP-based density $q_\theta$ has the following closed-form:
\begin{align}
\qs{\theta}((0,0)|x)&=\f(S_{\{(0,0)\}|x;\theta}|x)\label{eq:eg_prof_likelihood1}\\
\qs{\theta}((1,1)|x)&=\f(S_{\{(1,1)\}|x;\theta}|x)\label{eq:eg_prof_likelihood2}\\
\qs{\theta}((1,0)|x)&=
p_{rel}((1,0)|x)\eta_1(x;\theta)\I_1(x;\theta,p)\notag\\
&\qquad\qquad\qquad+\eta_2(x;\theta)\I_2(x;\theta,p)+ 
\eta_3(x;\theta)\I_3(x;\theta,p),
\label{eq:eg_prof_likelihood4}
\end{align}
where 
\begin{align}
	\I_1(x;\theta,p)&\equiv1\Big\{\eta_3(x;\theta)\le p_{rel}((1,0)|x)\eta_1(x;\theta)\le \eta_2(x;\theta)\Big\}\label{eq:def_Theta1}\\
	\I_2(x;\theta,p)&\equiv 1\Big\{	p_{rel}((1,0)|x)\eta_1(x;\theta)> \eta_2(x;\theta)\Big\}\label{eq:def_Theta2}\\
	\I_3(x;\theta,p)&\equiv 1\Big\{p_{rel}((1,0)|x)\eta_1(x;\theta)< \eta_3(x;\theta)\Big\}.\label{eq:def_Theta3}
\end{align}
\end{proposition}
In \eqref{eq:eg_prof_likelihood4}, $q_\theta$'s form depends on the relative frequency $p_{rel}((1,0)|x)$ of outcome (1,0).
If $p_{rel}((1,0)|x)$ is in the interval $[\eta_3(x;\theta)/\eta_1(x;\theta),\eta_2(x;\theta)/\eta_1(x;\theta)]$ for $\theta\in\Theta_0$, $q_\theta((1,0)|x)$ is proportional to $p_{rel}((1,0)|x)$. If the relative frequency is outside the interval, $q_\theta((1,0)|x)$ equals either of the end points of the interval (i.e., $\eta_2$ or $\eta_3$). 
\end{example}

One can use Proposition \ref{prop:game_lfp} to test various hypotheses on $\theta$. The presence of strategic interaction effects can be examined by testing
	\begin{align}
	H_0:\beta^{(j)}=0,~j=1,2~~~v.s.~~~H_1:\beta^{(j)}<0,~\text{ for some } j.
	\end{align}
For this hypothesis, we let $\Theta_0=\{\theta=(\beta',\delta')':\beta=0\}.$\footnote{Further simplification of Step 2 (in Section \ref{ssec:overview}) is possible for this null hypothesis because $\eta_2(\theta;x)=\eta_3(\theta;x)$ if $\beta=0$ due to the completeness of the model under $H_0$. Hence, the LFP-based likelihood satisfies $\qs{\theta}(y|x)=\f(S_{\{y\}|x;\theta}|x)$ for all $y$.}
One can also construct confidence intervals for counterfactual probabilities.
Let 
\begin{align}
Y^{(j)}(x^{(j)},y^{(-j)})\equiv 1\{x^{(j)}{}{'}\delta^{(j)}+\beta^{(j)}y^{(-j)}+U^{(j)}\ge 0\}
\end{align}
be the \emph{potential entry decision} by player $j$ when the covariates and the other player's action are set to $(x^{(j)},y^{(-j)})$. The counterfactual entry probability is
\begin{align}
    \varphi(\theta)= P(Y^{(j)}(x^{(j)},y^{(-j)})=1)=F_\theta(\{u:x^{(j)}{}{'}\delta^{(j)}+\beta^{(j)}y^{(-j)}\ge-u^{(j)} \}),
\end{align}
where $F_\theta$ is the marginal distribution of $U$.
Letting $\Theta_0(\varphi^*)=\{\theta\in\Theta:\varphi(\theta)=\varphi^*\}$,
one can construct a confidence interval as in \eqref{eq:ci}.
 Corollary \ref{cor:univ_cs} ensures the finite sample coverage of this interval.\footnote{Alternatively, one may be interested in a \emph{counterfactual equilibrium outcome} $Y(x)$, which would realize when we set $X$ to $x$. Due to the incompleteness, the counterfactual equilibrium probability $P(Y(x)=y)$ is not unique, but its bounds are functions of $\theta$. For example, the sharp upper bound on $P(Y(x)=y)$ is $\varphi(\theta)
    =\f(G(U|x;\theta)\cap \{(1,0)\}\ne\emptyset)=\f(S_{\{(1,0)\}|x;\theta})+\f(S_{\{(0,1),(1,0)\}|x;\theta}).$} Various other functionals can also be handled. In related work, \cite{han2024setvalued} apply this paper's procedure to construct confidence intervals for average structural functions, average treatment effects, and the counterfactual probability of an individual's switching behavior.

\section{Theory behind the Robustness and Finite-sample Validity}\label{sec:background}
This section outlines the machinery behind Theorem \ref{thm:univ_inf} and how it is related to existing results. 

\subsection{Universal inference with Least Favorable Pair}
\cite{Wasserman:2020aa} consider a probabilistic model $\{P_\theta,\theta\in\Theta\}$, where each $P_\theta$ is a probability distribution over a measurable space $\mathcal Z$. Given a sample $Z_i\in \mathcal Z,i=1,\dots, n$, they construct split and cross-fit LR statistics with $\cL_0(\theta)=\prod_{i\in D_0}p_{\theta}(Z_i)$, where $p_\theta=dP_\theta/d\mu$ and show the their universal validity.\footnote{To be precise, they apply Markov's inequality to $T_n$, which can be viewed as bounding the tail probability using the moment generating function (MGF) of the log-likelihood. The key is to use a simple yet nontrivial result $E_{P_\theta}[e^{t\ln T_n}]\le 1$ for $t=1$. They call this approach ``poor man's Chernoff bound''. See their discussion on page 16882.} In their setting, the likelihood $\cL_0(\cdot)$ represents the conditional distribution of the observations $\{Z_i,i\in D_0\}$ over subsample $D_0$, given $D_1$. In our setting, the form of this distribution is unknown. Hence, instead of applying their argument directly, we construct $\cL_0$ using the theory of minimax tests.\footnote{See \cite{Lehmann:2006aa} (Ch. 8) for a general treatment of the topic.} 
 
 Let $P$ be the joint distribution of $Z=(Y,X)\in\cS$, $P_{Y|X}$ be the conditional distribution of $Y|X$, and $P_X$ be the marginal distribution of $X$.
For each $\theta\in\Theta$, the set of permissible distributions of $Z$
 is $\mathcal{P}_{\theta}=\{P\in \Delta(\cS):P_{Y|X}(\cdot|x)\in \cP_{\theta,x},x\in \cX\}$.

 For $\theta_0, \theta_1 \in \Theta$ such that $\mathcal P_{\theta_0}$ and $\mathcal P_{\theta_1}$ are disjoint, consider testing $H_0: \theta=\theta_0$, against $H_1:\theta=\theta_1$. 
For any test $\phi: \cS \mapsto [0,1]$, its rejection probability under $P$ is $E_{P}[\phi(Z)]=\int \phi(z)dP.$
Let $\pi_{\theta}(\phi)\equiv\inf_{P \in \mathcal{P}_{\theta}}E_P [\phi(Z)]$ be the \emph{power guarantee} of $\phi$ at $\theta$. This is the power value certain to be obtained regardless of the unknown selection mechanism. A test $\phi$ is a \emph{level-$\alpha$ minimax test} if it satisfies the following conditions:
\begin{equation}
	\sup_{P \in \mathcal{P}_{\theta_0}}E_P[\phi(Z)]\le \alpha~,
\label{eq:size} \end{equation}
and
\begin{equation}
	\pi_{\theta_1}(\phi)\ge \pi_{\theta_1}(\tilde\phi),~\forall \tilde \phi \text{ satisfying }\eqref{eq:size}.
\label{eq:power} \end{equation}
That is, the test maximizes the power guarantee subject to the uniform size control constraint (over $\mathcal P_{\theta_0}$).

Define the lower envelope of $\cP_\theta$ by
\begin{align}
 \nu_\theta(A)\equiv\inf_{P \in \mathcal{P}_\theta}P(A),~A\in \cC.
\label{eq:defnutheta} \end{align}
Define the correspondence
$\Gamma(x,u;\theta)\equiv\{(y,x)\in\cS: y\in G(u|x;\theta)\}$.
By Choquet's theorem \citep[e.g.,][]{Choquet1954,philippe1999decision,Molchanov:2006aa}, we may represent $\nu_\theta$ using the distribution of the random set $\Gamma(X,U;\theta)$:
\begin{align*}
    \nu_\theta(A)&=\int_\cX\int_\cU1\{\Gamma(x,u;\theta)\subset A\}dF_\theta(u|x)dP_X(x).
\end{align*}
The set function $\nu_\theta$ belongs to a class of two-monotone capacities whose properties have proven powerful for conducting robust inference \citep{Huber:1981aa}.\footnote{See Appendix \ref{sec:capacities}. $\nu_\theta$ is called a \emph{belief function} or a totally monotone capacity. The total monotonicity of $\nu_\theta$ follows from \cite{philippe1999decision} (Theorem 3). The foundations of belief functions are given by \cite{dempster1967upper} and \cite{shafer1982belief}. See \cite{Gul:2014aa} and \cite{epstein2015exchangeable} for the axiomatic foundations of the use of belief functions in incomplete models.} For this class, the rejection region of a minimax test takes the form $\{z:\Lambda(z)>t\}$ for a measurable function $\Lambda$. Furthermore, we may obtain the following result \citep[see][]{Huber:1973aa,kaido2019robust}. 

\begin{proposition}\label{prop:cs}	
Suppose $\cP_{\theta_0}\cap \cP_{\theta_1}=\emptyset$.
	Then, (i) there exists a \emph{least-favorable pair (LFP)} $(Q_{0},Q_{1})\in\mathcal P_{\theta_0}\times\mathcal P_{\theta_1}$ of distributions such that for all $t\in\mathbb R$,
	\begin{align}
	\sup_{P \in \mathcal{P}_{\theta_0}}P(\Lambda>t)&=Q_{0}(\Lambda>t)\\
	\inf_{P \in \mathcal{P}_{\theta_1}}P(\Lambda>t)&=Q_{1}(\Lambda>t),
	\end{align}
where $\Lambda$ is a version of $dQ_1/dQ_0.$

\noindent
(ii) The conditional densities $q_j(\cdot|x),j=0,1$ associated with $Q_j,j=0,1$ solve the following convex program:
\begin{align}
	(q_{0}(\cdot|x),q_{1}(\cdot|x))=\argmin_{\tilde q_{0}(\cdot|x),\tilde q_{1}(\cdot|x)}&~\sum_{y\in \cY}\ln\Big(\frac{\tilde q_0(y|x)+\tilde q_1(y|x)}{\tilde q_0(y|x)}\Big)(\tilde q_0(y|x)+\tilde q_1(y|x))\label{eq:lfp_program}\\
	s.t.&~ \nu_{\theta_0}(B|x)\le \sum_{y\in B}\tilde q_0(y|x),~B\subset \cY,~ x\in \cX\notag\\
	&~ \nu_{\theta_1}(B|x)\le \sum_{y\in B}\tilde q_1(y|x),~B\subset \cY,~ x\in \cX.\notag
\end{align}

\noindent
(iii) There exist constants $(\gamma,C)$ such that
\begin{align*}
		\phi(z) = \left\{
		\begin{array}{ll}
			1 & \text{ if} \quad \Lambda(z)>C\\
			\gamma & \text{ if} \quad \Lambda(z)=C\\
			0 & \text{ if} \quad \Lambda(z)<C,\\
		\end{array}
		\right.
\end{align*} 
is a level-$\alpha$ minimax test.
\end{proposition}
Heuristically, $Q_0$ is the probability distribution consistent with the null parameter value, which makes size control most difficult. $Q_1$ is the distribution consistent with the alternative parameter value, which is least favorable for maximizing power. The LR-test based on this pair is minimax optimal.\footnote{Proposition \ref{prop:cs} is a special case of Theorem 3.1 in  \cite{kaido2019robust}. They provide explicit forms of the minimax tests for several examples.}

We use Proposition \ref{prop:cs} in the proof of Theorem \ref{thm:univ_inf} (see Appendix \ref{ssec:proofs}). 
The following is a sketch of the proof arguments. Let $\theta_1\in \Theta$ and let $Q_{\theta_1}\in\cP_{\theta_1}$. 
For each 
 $\theta\in\Theta_0$, consider distinguishing $\cP_{\theta_1}$ against a singleton set $\{Q_{\theta_1}\}$ and
form the LFP $(Q_{\theta},Q_{\theta_1})$. Repeating this exercise across $\theta\in\Theta_0$, we form a null parametric model $\{q_{\theta},\theta\in\Theta_0\}$ with $q_{\theta}=dQ_{\theta}/d\mu$. 
Let $T_n(\theta_1)=\cL_0(\theta_1)/\cL_0(\hat\theta_0)$. Suppose $P^n\in \cP^n_{\theta}$ for some $\theta\in\Theta_0$. Then,
\begin{align}
P^n(T_n(\theta_1)>\frac{1}{\alpha})\le \alpha E_{P^n}\Big[\frac{\cL_0(\theta_1)}{\cL_0(\hat\theta_0)}\Big]\le \alpha \sup_{\tilde P^n\in  \mathcal P^n_{\theta}}E_{\tilde P^n}\Big[\frac{\cL_0(\theta_1)}{\cL_0(\theta)}\Big],\label{eq:heuristics}
\end{align}
The first inequality is due to Markov's inequality, and the second inequality is due to $\cL_0(\hat\theta_0)\ge \cL_0(\theta)$ by $\hat\theta_0$ being the RMLE.
We show that the supremum on the right-hand side of \eqref{eq:heuristics} is attained by the product of the least-favoarable distribution $Q_\theta$ at $\theta$ and use this result to show $ \sup_{\tilde P^n\in  \mathcal P^n_{\theta}}E_{\tilde P^n}\Big[\frac{\cL_0(\theta')}{\cL_0(\theta)}\Big]\le 1$. This ensures that the size of the test is at most $\alpha$. The argument above fixed $\theta_1$, but an estimator $\hat\theta_1$ is used in practice. In the proof of Theorem \ref{thm:univ_inf}, we use a conditioning argument to handle the randomness of $\hat\theta_1$. 

\begin{remark}\label{rem:approximation}
Our analysis requires $\cP_{\theta,x}$ to be the core of a two-monotone capacity.  In some applications, one may want to work with a set $\mathcal M_{\theta,x}$ of probabilities satisfying certain inequalities
\begin{align}
Q(A|x)\ge \kappa_\theta(A|x), ~A\in\mathcal A	
\end{align}
for some $\kappa_\theta(\cdot|x)$ and a class of events $\mathcal A\subset \mathcal C$. This may arise when one wants to work with non-sharp bounds or a subset of Artstein's inequalities for computational reasons. The set function $\kappa_\theta$ may be a containment functional or another function that gives a more conservative bound.
In such settings, one could apply our method after approximating the lower envelope of $\mathcal M_{\theta,x}$ by a containment functional \citep{Montes:2018aa,MONTES2018181}.
\end{remark}

\section{Monte Carlo experiments}

We examine the performance of the proposed test through simulations. First, we use the two-player entry game example with the following payoff:
	\begin{align}
	\pi^{(j)}=y^{(j)}\big(\theta^{(j)}y^{(-j)}+u^{(j)}\big),~j=1,2.\label{eq:mc_payoff}
	\end{align}
We then test $H_0:\theta^{(j)}=0,j=1,2$ against $H_1:\theta^{(j)}<0$ for some $j$. As discussed in Example \ref{ex:game1}, the model is complete under the null hypothesis, which determines $\cL_0$. Hence, one only needs to determine $\hat\theta_1$. We consider two options. Both are extreme estimators. The first one is a minimizer of a sample criterion function based on the sample analog of the sharp identifying restrictions as in \eqref{eq:cht}. We call this estimator a moment-based estimator.
 The second one maximizes the information-based objective function as in \eqref{eq:klic_proj}.
We call this estimator an MLE.

We set the sample sizes to $n\in\{50,100,200\}$ and calculate the rejection probability of the tests at the alternatives with $\theta^{(j)}=-h$ with $h\ge 0$. Under each alternative, the outcome $Y_i=(1,0)$ is selected with probability 0.5 whenever the model admits multiple equilibria.
Table \ref{tab:power1} reports the size and power of the two cross-fit tests. The two tests have similar power profiles, suggesting the choice of the initial estimator $\hat\theta_1$ does not seem to matter, at least for this example. 
Overall, the proposed tests' rejection probabilities under $H_0$ are nearly 0. However, they exhibit considerable power even in small samples as $h$ deviates from 0. They could detect mild strategic interaction effects (e.g., $h=0.5$, which is a half standard-deviation unit of $u^{(j)}$) with high rejection probabilities, even in small samples. 

\begin{table}[htbp]
    \centering
    \caption{Size and Power of the Cross-Fit Tests for testing $H_0:\theta^{(j)}=0,j=1,2$}
    \label{tab:power1}
    \resizebox{\textwidth}{!}{	
    \begin{tabular}{lccccccccccccccc}
        \toprule
       & Size & \multicolumn{14}{c}{Power (values of $h$ below)} \\
        
        \cmidrule{3-16}
        &  & 0.069 & 0.138 & 0.207 & 0.276 & 0.345 & 0.414 & 0.483 & 0.552 & 0.621 & 0.690 & 0.759 & 0.828 & 0.897 & 0.966 \\
        \hline
        Panel A: ($n=50$) & \multicolumn{15}{c}{}\\
        LR-test (MLE $\hat\theta_1$) & 0 & 0.000 & 0.002 & 0.013 & 0.049 & 0.099 & 0.175 & 0.267 & 0.387 & 0.521 & 0.631 & 0.744 & 0.821 & 0.887 & 0.928 \\
        LR-test (moment-based $\hat\theta_1$) & 0 & 0.002 & 0.006 & 0.018 & 0.056 & 0.103 & 0.184 & 0.280 & 0.391 & 0.517 & 0.631 & 0.755 & 0.825 & 0.891 & 0.933 \\
        & \multicolumn{15}{c}{}\\
        Panel B: ($n=100$) & \multicolumn{15}{c}{}\\
        LR-test (MLE $\hat\theta_1$) & 0 & 0.001 & 0.011 & 0.073 & 0.196 & 0.370 & 0.576 & 0.760 & 0.885 & 0.948 & 0.973 & 0.992 & 0.996 & 1.000 & 1.000 \\
        LR-test (moment-based $\hat\theta_1$) & 0 & 0.002 & 0.016 & 0.081 & 0.209 & 0.383 & 0.582 & 0.762 & 0.877 & 0.942 & 0.976 & 0.988 & 0.993 & 0.998 & 1.000 \\
        & \multicolumn{15}{c}{}\\
        Panel C: ($n=200$) & \multicolumn{15}{c}{}\\
        LR-test (MLE $\hat\theta_1$)  & 0.001 & 0.007 & 0.060 & 0.235 & 0.522 & 0.794 & 0.948 & 0.988 & 0.997 & 1 & 1 & 1 & 1 & 1 & 1 \\
        LR-test (moment-based $\hat\theta_1$) & 0.002 & 0.008 & 0.066 & 0.246 & 0.521 & 0.776 & 0.940 & 0.988 & 0.996 & 1 & 1 & 1 & 1 & 1 & 1 \\
    
        \bottomrule
    \end{tabular}}
\end{table}

\begin{table}[htbp]
    \centering
    \caption{Size and Power of the Cross-Fit Tests for testing $H_0:\delta^{(j)}=0,j=1,2$.}
    \label{tab:power2}
    \resizebox{\textwidth}{!}{	
    \begin{tabular}{lccccccccccccccc}
        \toprule
     $n$  & Size & \multicolumn{14}{c}{Power (values of $h$ below)} \\
               \cmidrule{3-16}
    &  & 0.105 & 0.211 & 0.316 & 0.421 & 0.526 & 0.632 & 0.737 & 0.842 & 0.947 & 1.053 & 1.158 & 1.263 & 1.368 & 1.474 \\
    \hline
       50 &   0 & 0.000 & 0.005 & 0.030 & 0.100 & 0.222 & 0.354 & 0.493 & 0.612 & 0.732 & 0.834 & 0.894 & 0.940 & 0.972 & 0.982 \\
       100 &  0 & 0.000 & 0.011 & 0.070 & 0.194 & 0.375 & 0.499 & 0.631 & 0.748 & 0.864 & 0.926 & 0.963 & 0.978 & 0.989 & 0.998 \\
        200 & 0 & 0.000 & 0.056 & 0.252 & 0.504 & 0.665 & 0.790 & 0.867 & 0.929 & 0.965 & 0.981 & 0.987 & 0.991 & 0.994 & 0.997 \\
        300 & 0 & 0.006 & 0.158 & 0.558 & 0.809 & 0.912 & 0.954 & 0.973 & 0.989 & 0.998 & 0.996 & 0.996 & 0.998 & 0.997 & 0.997 \\
         \bottomrule
    \end{tabular}}
    \bigskip

     \footnotesize{Note: The size and power are calculated based on $S=1000$ simulations.  The DGP involves two independently distributed covariates $(X^{(1)},X^{(2)})$ with $|supp(X^{(j)})|=5$. The average sample size in each bin  is at most 12.}
\end{table}

In the second experiment, we use the specification of payoff functions in \eqref{eq:payoff}. For each $j$, $X^{(j)}$ is a covariate that takes $K=5$ discrete values $\{-2,-1,0,1,2\}$, and the two covariates are generated independently. When multiple equilibria are predicted, one of them gets selected with a probability of 0.5.
We test $H_0:\delta^{(j)}=0,j=1,2$ against $H_1:\delta^{(j)}\ne 0$ for some $j$. A notable difference from the previous specification is that the model is incomplete under the null hypothesis, and $\theta$ contains nuisance components (i.e., $\beta^{(j)},j=1,2$ for this design). 
We evaluate the power of the test against alternatives with $\delta^{(j)}=h,h>0$ for samples of size $n\in\{50,100,200,300\}$. For this experiment, we use the moment-based estimator as $\hat\theta_1$.
Table \ref{tab:power2} summarizes the result.
As in the previous design, the test's size is nearly 0, but it has meaningful power against alternatives, even in small samples. The test also exhibits monotonically increasing power curves. Hence, it again shows promise in detecting violations of the restrictions.

In the third experiment, we compare our test to an existing test of \cite{BCS} using large samples ($n=5000,7500$).\footnote{With covariates taking 25 different values with equal probabilities, the average sample size in each bin would become extremely small if we use a small sample (e.g., $n=100$). This feature was not an issue for the cross-fit LR test. However, it caused computational issues for implementing \citeposs{BCS} test. Instead of adding modifications not considered in their original paper, we work in an environment in which their procedure is reliable. We found that their test was computationally reliable for sample sizes $n=5000,7500$.}   We call their procedure a moment-based test. Figure \ref{fig:power_bcs_comparison} reports the power curves of the two tests. Both tests show monotonically increasing power curves, and neither of them is uniformly dominant. This result shows that the cross-fit LR test can compete with a well-established test in terms of power.
Table \ref{tab:comptime} reports the computation time required to implement the tests.\footnote{They were computed using Boston University's Shared Computing Cluster (SCC) nodes equipped with 2.6 GHz Intel Xeon processor (E5-2650v2) and 128GB memory (per node). } \citeposs{BCS} test uses a bootstrap critical value.
To mimic a realistic scenario, we parallelized their bootstrap replications ($B=500$) across multiple processors. The table shows that the cross-fit LR test takes about 14 seconds (without any parallelization), which is significantly below the computation time required for the moment-based test with parallelization.

In sum, the simulation results show that the cross-fit LR test has considerable power, even in small samples.
In large samples for which existing tests are applicable, the proposed test has power properties comparable to a well-established test.

\begin{figure}
    \centering
    \includegraphics[scale=.7]{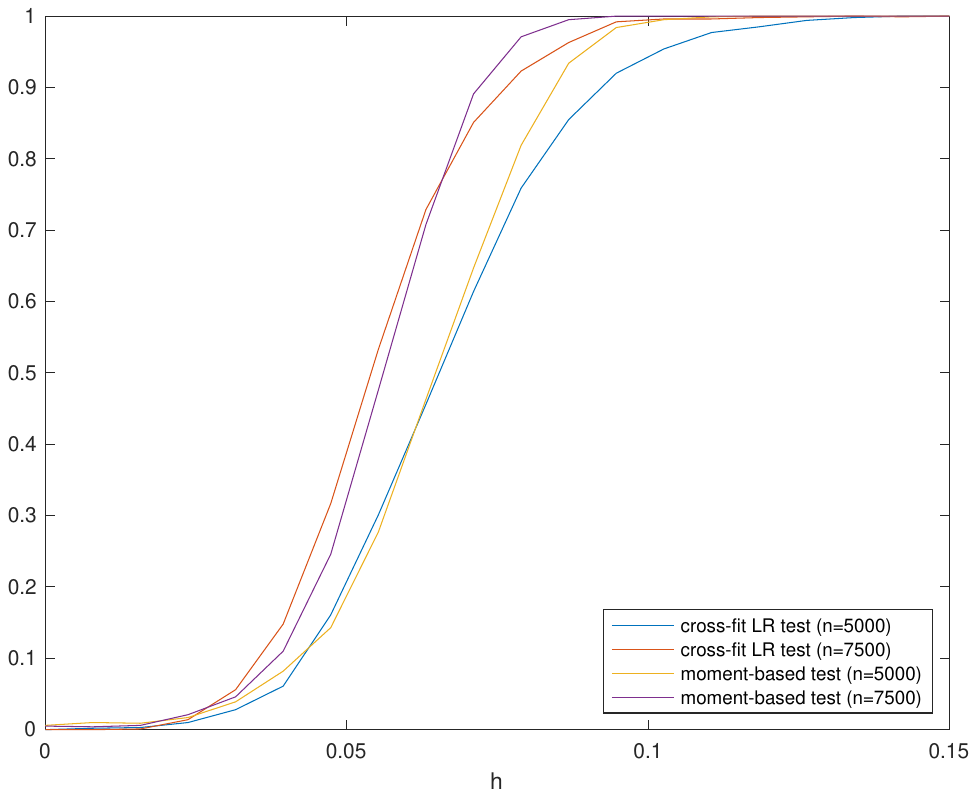}
    \caption{Power of the Cross-fit LR and Moment-based Tests: ($S=1000$ replications)}
    \label{fig:power_bcs_comparison}
\end{figure}

\begin{table}
    \centering
    \begin{tabular}{cccc}
        \toprule
        Cross-fit LR test & \multicolumn{3}{c}{Moment-based test} \\
        \cmidrule(lr){2-4}
        & 4 cores & 8 cores & 16 cores \\
        \midrule
        13.75 &  111.65 &  56.64 &   41.84 \\
        \bottomrule
    \end{tabular}
    
    \caption{Computation Time (in seconds)}
    \label{tab:comptime}
    \footnotesize{Note: The median computation time is calculated based on $S=1000$ simulations for the Cross-fit LR test. For the moment-based test of \cite{BCS}, we parallelized bootstrap replications with 4, 8, and 16 cores. The median computation time is calculated based on $S=100$ simulation repetitions. }
\end{table}

\section{Discussion}
This paper develops a novel likelihood-based test and confidence sets for incomplete models. They apply to a wide range of discrete choice models involving set-valued predictions. Yet, they are simple to implement. To retain simplicity, this paper uses simple two-fold cross-fitting. An avenue for further research is to examine whether alternative sample-splitting schemes can improve the statistical properties and replicability of the proposed method. For the latter, \cite{Ritzwoller:2024aa} recently proposed a way to control an error rate due to sample splitting by sequentially aggregating statistics, which is a fruitful direction.
Our simulation results suggest that the proposed test's power is comparable to that of a well-established moment-based test in large samples. It is worthwhile to study the proposed test's asymptotic power properties further.

\clearpage
\appendix

\section{Proofs}

\subsection{Preliminaries}\label{sec:capacities}
We introduce set functions called capacities and discuss their basic properties. We refer to \cite{Denneberg:1994aa} for technical treatments.

Let $\Sigma_\cS$ be the Borel $\sigma$-algebra. A function $\nu:\Sigma_\cS\to\mathbb R$ with $v(\emptyset)=0$ is a \emph{capacity}. Throughout, we assume $\nu(A)\ge 0,~\forall A\in\Sigma_\cS$, $\nu(\cS)=1$ (i.e. normalized). We also assume $\nu$ is monotone. That is, for any $A,B\in\Sigma_{\cS}$, $A\subseteq B\Rightarrow \nu(A)\le \nu(B)$.  Capacity $\nu$ is said to be \emph{monotone of order $k$} or, for short, \emph{k-monotone} if for any $A_i\subset S,i=1\cdots,k$,
\begin{align}
	\nu\big(\cup_{i=1}^k A_i\big) \ge \sum_{I\subseteq\{1,\cdots,k\}, I\ne \emptyset}(-1)^{|I|+1}\nu\big(\cap_{i\in I}A_i\big).
\end{align}
If the property holds for any $k$, it is called a \emph{totally monotone} capacity.
The conjugate $\nu^*(A)=1-\nu(A^c)$ of a $k$-monotone capacity is called a \emph{$k$-alternating} capacity.
For any capacity $\nu$ and a real-valued function $f$ on $\cS$, the \emph{Choquet integral} of $f$ with respect to $\nu$ is defined by
\begin{align}
\int fd\nu\equiv \int_{-\infty}^0(\nu(\{s:f(s)\ge t\})-\nu(\cS))dt+\int_0^\infty \nu(\{s:f(s)\ge t\})dt.\label{eq:choquet}
\end{align}

Let $A(D_1)\in \Sigma_\cS$ be a measurable set, which is allowed to depend on subsample $D_1$. 
For each $\theta\in \Theta$, let
\begin{align}
	\nu^*_\theta(A(D_1)|D_1)\equiv \int1\{G(u|X;\theta)\cap A(D_1)\ne\emptyset\}dF_\theta(u).
\end{align}
The set function $\nu^*_\theta(\cdot|D_1)$ is then a totally-alternating capacity \citep{philippe1999decision}.

\subsection{Proof of Theorems}\label{ssec:proofs}
\begin{proof}[Proof of Theorem \ref{thm:univ_inf}]
We present a version of the proof for the split-sample statistic $T_n$ first. Let $\theta\in\Theta_0$ and
let $P^n\in \mathcal P^n_{\theta}$. Let $P^{D_j}$ be the marginal distribution of $P^n$ on $D_j$, and let $P^n(\cdot|D_1)$ be the conditional distribution given $\{(Y_i,X_i),i\in D_1\}$. 
By Markov's inequality,
\begin{align}
	P^n(T_n>\frac{1}{\alpha})\le \alpha E_{P^n}[T_n]=\alpha E_{P^{D_1}}[E_{P^n}[T_n|D_1]]. \label{eq:sizeproof1}
\end{align}
For each $i\in D_0$, let $\Lambda_i=q_{\hat\theta_1}(Z_i)/q_\theta(Z_i)$ and define $T^*_n(\theta)=\cL_0(\hat\theta_1)/\cL_0(\theta)=\prod_{i\in D_0}\Lambda_i.$
We may bound $E_{P^{n}}[T_n|D_1]$ as follows.
\begin{align}
	E_{P^n}[T_n|D_1]&\le \sup_{\tilde P^n\in  \mathcal P^n_{\theta}}E_{\tilde P^n}[T_n|D_1]\notag\\
	&\stackrel{(1)}{\le}\sup_{\tilde P^n\in  \mathcal P^n_{\theta}}E_{\tilde P^n}[T^*_n(\theta)|D_1]\notag\\
	&\stackrel{(2)}{=}\prod_{i\in D_0}\sup_{\tilde P\in \mathcal P_\theta}E_{\tilde P}[\Lambda_i|D_1]\notag\\
	&\stackrel{(3)}{=}\prod_{i\in D_0}\int \Lambda_id\nu^*_\theta(\cdot|D_1),\label{eq:sizeproof2}
\end{align}	
where (1) is due to $T_n\le T^*_n(\theta)$ for any $\theta\in\Theta_0$ due to $\hat\theta_0$ being the maximizer of $\cL_0$, (2) follows because of Assumption \ref{as:iidF} (i) and the definition of $\cP^n_\theta$ in \eqref{eq:cPtheta}, and (3) follows from $\sup_{\tilde P\in \mathcal P_\theta}E_{\tilde P}[\Lambda_i|D_1]=\int \Lambda_id\nu^*_\theta(\cdot|D_1)$, a property of the Choquet integral for two-alternating capacities \citep{Schmeidler:1986aa}.
Observe that 
\begin{align}
	\int \Lambda_id\nu^*(\cdot|D_1)=\int_0^\infty \nu^*_\theta(\Lambda_i\ge t|D_1)dt=\int_0^\infty \nu^*_\theta(\Lambda_i> t|D_1)dt.\label{eq:sizeproof3}
\end{align}
Suppose $\cP_{\theta}\cap\cP_{\hat\theta_1}=\emptyset$.
By Proposition \ref{prop:cs} (i), there exists $Q_\theta\in \mathcal P_\theta$  such that, for all $t$,
\begin{align}
	\nu^*_\theta(\Lambda_i> t|D_1)=Q_\theta(\Lambda_i> t|D_1),\label{eq:sizeproof4}
\end{align}
and $Q_\theta$'s conditional density $q_\theta$ solves \eqref{eq:def_qtheta} by Proposition \ref{prop:cs} (ii). If $\cP_{\theta}\cap\cP_{\hat\theta_1}$ is non-empty, we take $Q_\theta$ to be $Q_{\hat\theta_1}$.

Let $A$ be the support of $q_\theta$. Then,
\begin{multline}
	\int_0^\infty \nu^*_\theta(\Lambda_i> t|D_1)dt=\int_0^\infty Q_\theta(\Lambda_i> t|D_1)dt\\
 =E_{Q_\theta}\Big[\frac{q_{\hat\theta_1}(Z_i)}{q_{\theta}(Z_i)}\Big|D_1\Big]=\int_A\frac{q_{\hat\theta_1}(z)}{q_{\theta}(z)}q_\theta(z)dz\le \int q_{\hat\theta_1}(z)dz=1.\label{eq:sizeproof5}
\end{multline}
Combining \eqref{eq:sizeproof2}-\eqref{eq:sizeproof5} yields $E_{P^{D_1}}[T_n|D_1]\le 1.$ Conclude that $P^n(T_n>\frac{1}{\alpha})\le \alpha E_{P^n}[T_n]\le 1$ from \eqref{eq:sizeproof1} and observe that the bound applies uniformly across $P^n\in\cP^n_0$.

For the cross-fit test, arguing as in \eqref{eq:sizeproof1},
\begin{align}
	P^n(S_n>\frac{1}{\alpha})\le \alpha E_{P^n}[S_n]=\alpha E_{P^n}\Big[\frac{T_n+T^{\text{swap}}_n}{2}\Big]=\frac{\alpha}{2}(E_{P^n}[T_n]+E_{P^n}[T^{\text{swap}}_n]).
\end{align}
 The rest of the proof is essentially the same. 
\end{proof}

\begin{proof}[Proof of Corollary \ref{cor:univ_cs}]
The argument is by the standard test inversion. By Theorem \ref{thm:univ_inf}, for any $n$,
\begin{align*}
    P^n(\varphi^*\notin  CS_n)=P^n(S_n(\varphi^*)>\frac{1}{\alpha})\le \alpha,
\end{align*}
implying $ P^n(\varphi^*\in  CS_n)  \ge 1-\alpha$ uniformly in  $(\varphi^*,P^n)\in \mathcal F^n$.
\end{proof}

\subsection{Proof of the Propositions}

\begin{proof}[Proof of Proposition \ref{prop:game_lfp}]
By \eqref{eq:frak_q_CT}, it suffices solve the following program
\begin{align*}
\argmin_{q(\cdot|x)\in \Delta^\cY}&~\sum_{y\in \{(0,0),(1,1),(1,0),(0,1)\}}\ln\Big(\frac{q(y|x)+p(y|x)}{q(y|x)}\Big)(q(y|x)+p(y|x))\\
	s.t.&~q((0,0)|x)=\f(S_{\{(0,0)\}|x;\theta}|x)\\
 &~q((1,1)|x)=\f(S_{\{(1,1)\}|x;\theta}|x)\\
	 &\f(S_{\{(1,0)\}|x;\theta}|x)\le q((1,0)|x)\leq \f(S_{\{(1,0)\}|x;\theta})+\f(S_{\{(0,1),(1,0)\}|x;\theta}|x).
\end{align*}
Note that $q_\theta((0,0)|x)$ and $q((1,1)|x)$ are determined by the equality constraints, which gives \eqref{eq:eg_prof_likelihood1}-\eqref{eq:eg_prof_likelihood2}. Therefore, it suffices to solve the problem above for $q((1,0)|x)$. Let $z=q((1,0)|x)$ and note that one can express $q((0,1)|x)$ as $q((0,1)|x)=\eta_1(\theta;x)-z$. Consider 
\begin{align*}
\min_{z\in [0,1]} &  \ln\Big(\frac{z+p((1,0)|x)}{z}\Big)(z+p((1,0)|x))\\
&~~+\ln \Big(\frac{\eta_1(\theta;x)-z+p((0,1)|x)}{\eta_1(\theta;x)-z}\Big)(\eta_1(\theta;x)-z+p((0,1)|x))\\
s.t.& \f(S_{\{(1,0)\}|x;\theta}|x)\le z\leq \f(S_{\{(1,0)\}|x;\theta})+\f(S_{\{(0,1),(1,0)\}|x;\theta}|x)
\end{align*}
If $\f(S_{\{(0,1),(1,0)\}|x;\theta}|x)>0$, Slater's condition is satisfied. Solving the Karush-Kuhn-Tucker (KKT) condition for this problem yields \eqref{eq:eg_prof_likelihood4}.  If $\f(S_{\{(0,1),(1,0)\}|x;\theta}|x)=0$, the model is complete, and the solution reduces to $q_\theta(y|x)=\f(S_{\{y\}|x;\theta}|x)$ for all $y$, which is a special case of \eqref{eq:eg_prof_likelihood1}-\eqref{eq:eg_prof_likelihood4}.
\end{proof}

\section{DGPs with Unknown Dependence}\label{sec:dependence}
We consider more general data-generating processes than $\mathcal F_0$. 
For each $i$, let $G_i:\cU\times\cX\times\Theta\to \cY$ be a weakly measurable correspondence.
Let
\begin{align}
G^n(u^n|x^n;\theta)\equiv \prod_{i=1}^n G_i(u_i|X_i;\theta).    
\end{align}
Let
\begin{multline}
	\tilde{\mathcal P}^n_\theta\equiv\Big\{P^n\in \Delta(\cS^n):P^n(\cdot|x^n)=\int_{\cU^n}\eta(A|u^n,x^n) dF^n_\theta(u^n),~\forall A\in \Sigma_{\cY^n},\\
	\eta(\cdot|u^n,x^n)\in \Delta(G^n(u^n|x^n;\theta)),~a.s.\Big\}.
\end{multline}
This set allows arbitrary dependence of the outcome sequence $Y^n=(Y_1,\dots,Y_n)$ through the selection mechanism across $n$ units.
The following result is from \cite{kaido2019robust}.
\begin{theorem}\label{thm:kz19}
	Suppose $\{(U_i,X_i)\}$ are independently distributed across $i$. Suppose, for each $i$, $\mathcal P_{\theta_0,i}$ and $\mathcal P_{\theta_1,i}$ are disjoint.
 Then, LFP $(Q_{0}^n,Q_{1}^n)\in\tilde{\mathcal P}_{\theta_0}^n\times\tilde{\mathcal P}_{\theta_1}^n$ exists such that for all $t\in\mathbb R_+$,
	\begin{align}
	\nu^{*,n}_{\theta_0}(\Lambda_n>t)&=Q_0^n(\Lambda_n>t)\\	
	\nu^{n}_{\theta_1}(\Lambda_n>t)&=Q_1^n(\Lambda_n>t),	
	\end{align}
where $\Lambda_n=dQ_1^n/dQ_0^n$. The LFP consists of the product measures:
\begin{align}
	Q^n_0=\bigotimes_{i=1}^n Q_{0,i},~~\text{ and }~~Q^n_1=\bigotimes_{i=1}^n Q_{1,i},
\end{align}
where, for each $i\in\mathbb N$, $(Q_{0,i},Q_{1,i})\in \mathcal P_{\theta_0,i}\times\mathcal P_{\theta_1,i}$ is the LFP in the $i$-th experiment;    
\end{theorem}
The result above states for distinguishing $\theta_0$ against $\theta_1$, the least favorable pair consists of product measures. If $(X_i,U_i)$ are identically distributed across $i$, the theorem implies the LFP consists of i.i.d. laws. This allows us to generalize the result in the text.

The i.i.d. sampling assumption in  \cite{Wasserman:2020aa}   can be relaxed as long as one can compute a likelihood for $D_0$ conditional on $D_1$. This requirement is crucial for the sample-splitting (and cross-fitting) to work. The distribution of outcomes can be heterogeneous and dependent across $i$ in unknown ways in this section. This feature makes it hard to define a conditional likelihood and apply the argument we used. 
For the conditioning argument to work, we construct $\hat\theta_1$ using outcomes that can be uniquely determined by $(X_i,U_i)$ but not by the selection.
Theorem \ref{thm:universal2} below generalizes the universal inference result to a wider class of distributions.

We now explain how to construct an estimator $\hat\theta_1$ of $\theta$ under the general model.
\begin{definition}[Initial estimator]\label{def:theta1}
$\hat\theta_1$ is an \emph{initial estimator} of $\theta$ such that
(i)	it is constructed from $\{W_i=\varphi(Y_i,X_i),i\in D_1\}$ for a measurable function $\varphi:\cY\times\cX\to\cW$; (ii) Under $H_0$, one can wirte $W_i$ as $W_i=f(X_i,U_i)$ for some measurable function $f:\cX\times\cU\to \cW$.
\end{definition}
Definition \ref{def:theta1} states $\hat\theta_1$ is a function of observable variables from $D_1$, which can also be expressed as a function of the exogenous variables $(X_i,U_i)$ at least under the null hypothesis. In Example \ref{ex:game1}, the model is complete under $H_0$. Hence, there is a unique reduced form $Y_i=g(U_i|X_i;\theta)$. This ensures we can construct $\hat\theta_1$ using directly $W_i=Y_i,i\in D_1$, i.e., $\varphi$ is the identity map. For example, one can use an extremum estimator
\begin{align}
	\hat\theta_1\in \argmin_{\theta\in\Theta}\hat{\mathsf Q}_1(\theta),
\end{align}
where $\hat{\mathsf Q}_1$ can be an objective function based on moment inequalities as discussed earlier.

In Example \ref{ex:game1}, the model is incomplete under $H_0$. Hence, some outcomes cannot be expressed as a function of the exogenous variables (e.g., $Y_i=(1,0)$ gets selected from multiple equilibria). Nonetheless, as shown by \cite{bresnahan1990entry,berry1992estimation}, the number of entrants
\begin{align}
	W_i=\varphi(Y_i)=\sum_{j} Y_i^{(j)}
\end{align}
is uniquely determined as a function of $(X_i,U_i)$. For a two-player game, a natural candidate is 
\begin{align}
	\hat\theta_1\in\argmax_{\theta\in\Theta}\sum_{i\in D_1}\sum_{k=0}^21\{W_i=k\}\ln F_\theta(S_{k,X_i}).
\end{align}
With this change, the rest of the procedure remains the same. For each $i$, let $q_{\hat\theta_1,i}$ be a solution to \eqref{eq:defp}. Solve \eqref{eq:def_qtheta} to find $q_{\theta,i},\theta\in\Theta_0$. Finally, form the cross-fit LR statistic as in \eqref{eq:def_Sn}. One can use the product LR statistic due to Theorem \ref{thm:kz19}. The following theorem establishes the validity of the procedure.
Let $\tilde{\cP}^n_0\equiv\{P^n\in\tilde{\cP}^n_\theta:\theta\in\Theta_0\}$.

\begin{theorem}\label{thm:universal2}
Suppose $(X_i,U_i),i=1,\dots,n$ are independently distributed across $i$, and $U_i|X_i=x\sim F_{\theta,i}(\cdot|x)$.
\begin{align}
\sup_{P^n\in \tilde{\cP}^n_0}P^n\big(S_n>\frac{1}{\alpha}\big)\le \alpha.      
\end{align}
  
\end{theorem}

\begin{proof}
Let $\theta\in\Theta_0$ and $\hat\theta_1$ be an estimator of $\theta$ described in Definition \ref{def:theta1}.  Note that $\hat\theta_1$ is a function of $(X_i,U_i),i\in D_1$ only, and $(X_i,U_i)$ are independently distributed across $i$. Therefore, conditioning on $(X_i,U_i),i\in D_1$ does not provide additional information on the observations from $D_0$ through selection across the two subsamples. Below, we condition on $D_1$ and treat $\hat\theta_1$ as fixed.
Let $Q_{\hat\theta_1}^n\in \tilde{\mathcal P}^n_{\hat\theta_1}$. Consider a minimax testing problem between $\tilde P^n_\theta$ and $\{Q_{\hat\theta_1}^n\}$.  By Theorem \ref{thm:kz19}, there is a product LFP $(Q_\theta^n,Q_{\hat\theta_1}^n)\in \tilde P^n_\theta\times\{Q_{\hat\theta_1}^n\}$. We use $Q_\theta^n$ below.

The proof of the theorem is the same as the proof of Theorem \ref{thm:univ_inf} up to \eqref{eq:sizeproof1}. We may bound $E_{P^{n}}[T_n|D_1]$ as follows.
\begin{align}
	E_{P^n}[T_n|D_1]&\le \sup_{\tilde P^n\in  \tilde{\mathcal P}^n_{\theta}}E_{\tilde P^n}[T_n|D_1]\notag\\
	&\le\sup_{\tilde P^n\in  \tilde{\mathcal P}^n_{\theta}}E_{\tilde P^n}[T^*_n(\theta)|D_1]\notag\\
	&=\sup_{\tilde P^n\in  \tilde{\mathcal P}^n_{\theta}}\int_0^\infty \tilde P^n\big(T^*_n(\theta)>t\big|D_1\big)dt\notag\\
 &\le \int_0^\infty \sup_{\tilde P^n\in  \tilde{\mathcal P}^n_{\theta}}\tilde P^n\big(T^*_n(\theta)>t\big|D_1\big)dt\notag\\
 &=\int_0^\infty  Q^n_{\theta}\big(T^*_n(\theta)>t\big|D_1\big)dt\notag,
\end{align}	   
where the last equality is due to $Q_\theta^n$ being the least-favorable distribution in $\tilde P^n_\theta$. Let $A^n$ be the support of $Q^n_{\theta}.$ Then,
\begin{multline}
\int_0^\infty Q^n_{\theta}\big(T^*_n(\theta)>t\big|D_1\big)dt=E_{Q^n_{\theta}}\left[\frac{\prod_{i\in D_0}q_{\hat\theta_1}(Z_i)}{\prod_{i\in D_0}q_{\theta}(Z_i)}\right]\\=\int_{A^n}\frac{\prod_{i\in D_0}q_{\hat\theta_1}(z_i)}{\prod_{i\in D_0}q_{\theta}(z_i)}
\prod_{i\in D_0}q_{\theta}(z_i)dz^{D_0}=\int_{A^n}\prod_{i\in D_0}q_{\hat\theta_1}(z_i)dz^{D_0}\le \prod_{i\in D_0}\int q_{\hat\theta_1}(z_i)dz_i=1.
\end{multline}
Note that the second equality uses the fact $\prod_{i\in D_0}q_{\theta}$ is the density of $Q_\theta^n$.
The rest of the argument is the same as the proof of Theorem \ref{thm:univ_inf}.
\end{proof}

\clearpage

\bibliographystyle{ecta}
% \bibliography{lrinferencejan30}

\end{document}